\RequirePackage{etex}
\documentclass{amsart}
\usepackage{amsmath, amsthm, amssymb,  mathrsfs, epsfig}
\usepackage{dcpic, pictex, pinlabel}

\begin{document}

\newtheorem{theorem}{Theorem}[section]
\newtheorem{lemma}[theorem]{Lemma}
\newtheorem{lctheorem}[theorem]{theorem}
\newtheorem{proposition}[theorem]{Proposition}
\newtheorem{sublemma}[theorem]{Sublemma}
\newtheorem{corollary}[theorem]{Corollary}
\newtheorem{conjecture}[theorem]{Conjecture}
\newtheorem{question}[theorem]{Question}
\newtheorem{problem}[theorem]{Problem}
\newtheorem*{claim}{Claim}
\newtheorem*{criterion}{Criterion}
\newtheorem*{main_theorem}{Theorem A}
\newtheorem*{second_theorem}{Theorem B}
\newtheorem*{cauchy_theorem}{Topological Cauchy-Schwarz inequality}
\newtheorem*{lemma_schema}{Lemma Schema}

\theoremstyle{definition}
\newtheorem{definition}[theorem]{Definition}
\newtheorem{construction}[theorem]{Construction}
\newtheorem{notation}[theorem]{Notation}
\newtheorem{convention}[theorem]{Convention}
\newtheorem*{explanation}{Explanation}
\newtheorem*{warning}{Warning}

\theoremstyle{remark}
\newtheorem{remark}[theorem]{Remark}
\newtheorem{example}[theorem]{Example}

\numberwithin{equation}{subsection}

\newcommand{\marginal}[1]{\leavevmode\marginpar{\tiny\raggedright#1}} 

\newcommand{\innerprod}[2]{\langle {#1} , {#2} \rangle}

\newcommand{\torus}{{\rm{torus}}}
\newcommand{\area}{{\rm{area}}}
\newcommand{\length}{{\rm{length}}}
\newcommand{\id}{{\rm{Id}}}
\newcommand{\R}{{\mathbb R}}
\newcommand{\RP}{{\mathbb R}P}
\newcommand{\Z}{{\mathbb Z}}
\newcommand{\Q}{{\mathbb Q}}
\newcommand{\T}{{\mathbb T}}
\newcommand{\J}{{\mathbb J}}
\newcommand{\M}{{\mathcal M}}
\newcommand{\gn}{\textrm{gn}}   
\newcommand{\ess}{{\mathcal S}}
\newcommand{\dee}{{\mathcal D}}
\newcommand{\Mdot}{{\dot\M}}
\newcommand{\Pdot}{{\dot{{\mathcal P}}}}
\newcommand{\Owe}{{\mathcal O}} 
\renewcommand{\H}{{\mathbb H}}
\newcommand{\C}{{\mathbb C}}
\newcommand{\bbar}{\overline}
\newcommand{\hhat}{\widehat}
\newcommand{\til}{\widetilde}
\newcommand{\mmin}{\text{min}}
\newcommand{\inc}{\text{inc}}
\newcommand{\Ric}{\text{Ric}}
\newcommand{\inv}{\text{inv}}
\newcommand{\cusp}{\text{cusp}}
\newcommand{\Xb}{\stackrel{\leftrightarrow}{X}}

\def\hline{\bigskip\hrule\bigskip}  
\def\nn#1{{\it [#1]}}       
\def\cN{{\mathcal N}}
\def\cO{{\bf O}}    
\def\Oh{{\bf O}}    

\def\tn{\textnormal}

\title{Complexity Classes as Mathematical Axioms}
\author{Michael H. Freedman}
\address{Microsoft Station Q \\ University of California \\
Santa Barbara, CA 93106}
\email{michaelf@microsoft.com}

\date{2/8/2009, Version 1.03}

\begin{abstract}
Complexity theory, being the metrical version of decision theory, has long been suspected of harboring undecidable statements among its most prominent conjectures.  Taking this possibility seriously, we add one such conjecture, $P^{\# P} \neq NP$, as a new ``axiom'' and find that it has an implication in 3-dimensional topology. This is reminiscent of Harvey Friedman's work on finitistic interpretations of large cardinal axioms.
\end{abstract}

\maketitle
\vspace{-0.12in}
\section{Introduction}

This short paper introduces a new subject with a simple example. The theory of computation defines a plethora of complexity classes. While the techniques of diagonalization and oracle relativization have produced important separation results, for nearly forty years the most interesting (absolute) separation conjectures, such as $P \neq NP$ remain unproven, and with the invention of ever more complexity classes, analogous separation conjectures have multiplied
in number.

With no prospect in sight for proving these conjectures (within
$ZFC$) and the suspicion that some are actually independent, we propose considering them instead as potential axioms and looking for what implications they might have in mathematics as a whole.  This program is analogous to the search for interesting ``finitistic'' consequences of large cardinal axioms, an area explored by Harvey Friedman and collaborators (e.g. \cite{Friedman}).
(Although, in the latter case, the large cardinal axioms are actually known to be independent of $ZFC$.)

What would be the best possible theorem in this subject?  It would be to postulate a very weak separation ``axiom,'' say $P \neq PSPACE$, and prove the Riemann hypothesis, i.e. an important mathematical result apparently far removed from complexity theory. Of course, we should be more modest.  We will assume a more technical but well accepted separation ``axiom'' $P^{\# P} \neq NP$, which we call Axiom A, and prove a theorem, Theorem A, in knot
theory. The theorem is easily and briefly expressed in terms of classical notions such as ``girth'' and ``Dehn surgery'' and appears to be as close to current research topics in knot theory as the known finitistic implications of the large cardinal axioms are to research in Ramsey theory, to continue that analogy.  Theorem A is extremely believable but seems to exist in a ``technique vacuum.'' What makes the theorem interesting is that it sounds both ``very
plausible'' and ``impossible to prove.''

\section{Theorem A}

We consider smooth links $L$ of finitely many components in $\R^3$ and their planar diagrams $D$. The girth of a diagram $D$ (in the $xz$-plane), $g(D)$, may be defined as the maximum over all lines $z =$ constant of the cardinality of the horizontal intersection $|D \cap (z = \text{constant})|$. For a link $L$, we define $girth(L)=\min\{g(D) | D \text{ is a diagram of } L\}$. Similarly, the complexity number $c(D)$ of a link diagram is defined as half its number of crossings plus half the number of local maxima and minima with respect to the $z$-coordinate. The complexity of a link, $c(L)=\min\{c(D)| D \text{ is a diagram of } L\}$. Theorem A addresses how girth can change under certain equivalence relations $\sim_r$ defined below.

Let $r \neq 6$ be an integer greater than or equal to $5$.  Consider passing from a link $L$ to $L \coprod U$, the disjoint union of $L$ and an additional unknotted component $U$, and then from $L \coprod U$ to $L'$ by performing $\frac{\pm 1}{4r}$-Dehn surgery on $U$. Denote by $\sim_r$ the equivalence relation on links generated by $L \rightarrow L'$.  In other words, this equivalence relation allows us to sequentially locate imbedded $2$-disks $\Delta$ transverse to $L$ and preform a $\pm 8\pi r$ twist across $\Delta$ to modify $L$; after several steps, we have arrived at a link, which we will denote $L'$, ``equivalent'' to $L$.  In slight abuse of notation, we also consider $\sim_r$ as an equivalence relation on diagrams: $D \sim_r D'$ iff $D$ represents $L$, $D'$ represents $L'$, and $L \sim_r L'$.

If $D$ and $D'$ are diagrams for the same link $L$, we may take their distance to be the minimum number of Reidemeister/Morse moves connecting $D$ to $D'$. Representative examples of these moves are displayed in figure \ref{reid}. We consider only diagrams in Morse position with respect to the $z$-coordinate and include in our count births, deaths, and level crossings, as well as the three familiar Reidemeister moves. Suppose next that $D$ and $D'$ are diagrams for $\sim_r$ equivalent links $L$ and $L'$. We need a measure of the distance between $D$ and $D'$.  It does not make sense to count each Dehn surgery as one step since the disk $\Delta$ may have an unboundedly complicated relation to $L$. There is no loss of generality, since $D$ can be modified by Reidemeister/Morse moves, in considering only disks $\Delta$ that meet $D$ in the standard form, seen in figure \ref{stdform} below.  More precisely, after Reidemeister/Morse moves, we may assume that in $D(L \coprod U)$, $U$ bounds a disk $\Delta$ and a neighborhood of $\Delta$ in $D(L \coprod U)$ appears as in figure \ref{stdform}.

Since $\pm 4r$-twisting along $\Delta$ introduces $4 r n(n-1)$ crossings, we will call half this, $2rn(n-1)$, the distance between the twisted and untwisted diagrams.  Now, $dist_r(D,D')$ can be defined to be the minimum number of (weighted) steps from $D$ to $D'$ where each isotopy induced, Reidemeister/Morse move is given weight $1$, except Reidemeister 1 which is weighted $\frac{3}{2}$ since three features can appear, a crossing, a local max and a local min, and each standard form $4r$-twist along $\Delta$ is given weight $2rn(n-1)$. (The exact form of $dist_r$ is irrelevant.  What is important is that if $D$ and $D'$ have a polynomial ``distance'' (in $\max(c(D), c(D'))$) then there is a polynomial sized certificate demonstrating that $L \sim_r L'$.  This clearly holds for $dist_r$ as defined.)

\begin{figure}[]
\labellist \small\hair 2pt
  \pinlabel $\text{Reidemeister }1$ at 67 98
  \pinlabel $\text{Reidemeister }2$ at 161 98
  \pinlabel $\text{Reidemeister }3$ at 294 98
  \pinlabel $\text{Birth}$ at 63 10
  \pinlabel $\text{Death}$ at 163 10
  \pinlabel $\text{Level Crossing}$ at 305 10
\endlabellist
\centering
\includegraphics[scale=1]{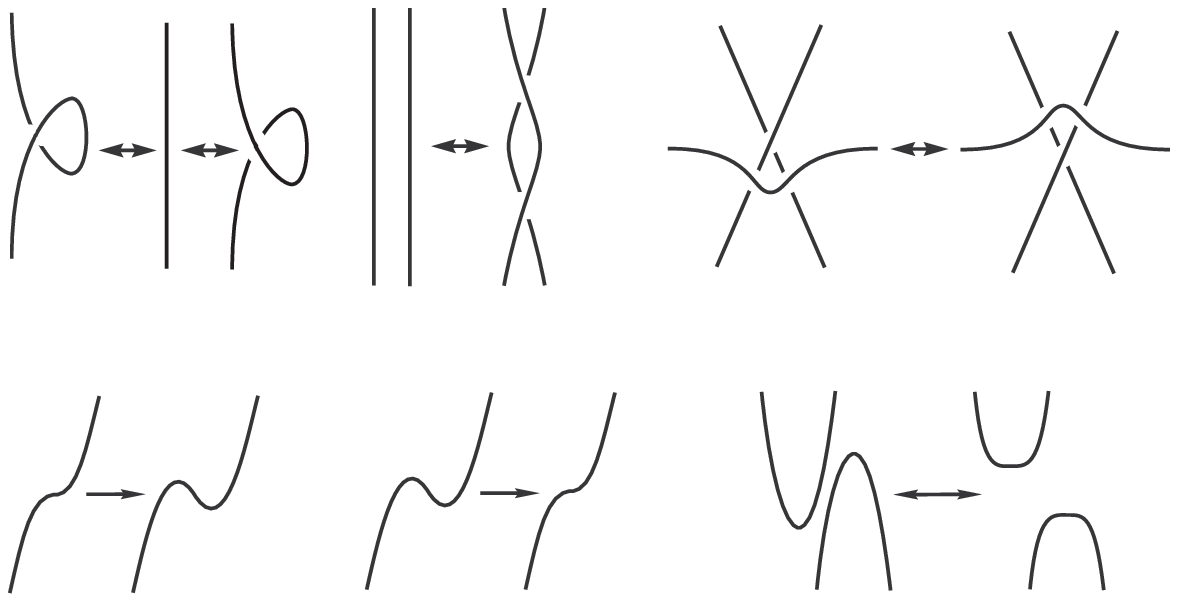}
\caption{}\label{reid}
\end{figure}

\begin{figure}[]
\labellist \small\hair 2pt
  \pinlabel $z$ at 38 55
  \pinlabel $x$ at 58 28
  \pinlabel $\Delta$ at 125 51
  \pinlabel $n$ at 125 17
  \pinlabel $\text{Part}$ at 217 59
  \pinlabel $\text{of \em{D}}(L\coprod U)$ at 217 51
\endlabellist
\centering
\includegraphics[scale=1.5]{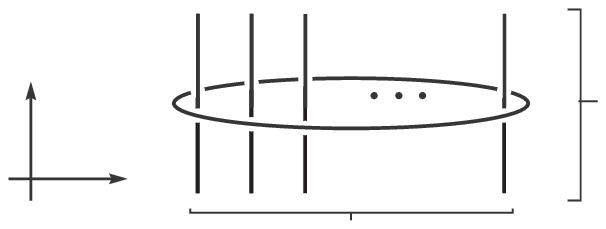}
\caption{} \label{stdform}
\end{figure}

\begin{main_theorem} \label{theorem_A}
If $r \geq 5$ is an integer not equal to $6$, $p$ a polynomial of
one variable, and $b,b'>0$ any constants, then there
exists a diagram $D$ such that if $D \sim_r D'$ then
\begin{align*}
g(D') > b\log(c(D))+b' \text{ unless} \\
dist_r(D, D') > p(c(D))
\end{align*}
\end{main_theorem}

Roughly, Theorem A says that some links $L$ cannot be made, via
$\sim_r$, extremely thin except possibly by an extraordinarily
elaborate sequence of moves.  It would be a surprise if the second
alternative actually occurred. In high dimensions \cite{Nabutovsky},
unsolvability of the triviality problem for groups implies that
geometric landscapes, for example that of the $5$-sphere in $S^6$,
are extremely (non-recursively) rough. However, this phenomenon has
not been seen in three manifold topology so it would be a surprise
if girth could be reduced only by a very long sequence of moves. We
conjecture that Theorem A remains true with the second alternative
omitted. However, for this statement no complexity axiom appears to
unlock the proof.

In the $1990$'s, A. Thompson \cite{Thompson} pointed out to me that girth,
by itself, can sometimes be computed exactly (see Claim below).  However, the equivalence relation $\sim_r$ is so disruptive of geometry that it appears to create the "technique vacuum" which we puncture with axiom A.

\begin{claim}
Let $k$ be the $(p,q)$-torus knot.  Then $g(k) = 2\min(p,q)$.
\end{claim}

\begin{proof}
So, $k \subset T \subset \R^3$, where $T$ is an unknotted torus
which we assume without loss of generality to be in generic (Morse)
position with respect to the $z$-coordinate of $\R^3$.  A
straightforward homological argument shows that some $z$-level must
intersect $T$ in one, in fact two, essential circles $C \coprod C'
\subset T$.  One easily builds imbedded disks (from bits of the
level plane and subsurfaces of $T$) $D$ and $D'$ with $\partial D =
D \cap T = C$ and $\partial D' = D' \cap T = C'$.  Thus, $C$ and
$C'$ are both meridians or both longitudes of $T$ and therefore must
contain at least $2\min(p,q)$ points of $k$.
\end{proof}

\section{A Complexity Reminder}

The exhibited inclusions in figure \ref{cplxincl} are all theorems
or tautologies.  The exhibited differences are all ``separation
conjectures'' to which we might grant the status of ``axioms.''
The existence of a problem $y \in P^{\# P} \setminus NP$ is the
axiom, ``Axiom A,'' we add to $ZF$, Zermelo-Fraenkel set theory, for the ``proof'' of Theorem A.

\begin{figure}[]
\labellist \small\hair 2pt
  \pinlabel $P$ at 85 53
  \pinlabel $NP$ at 84 72
  \pinlabel $P^{PP}=P^{\# P}$ at 166 61
  \pinlabel $PSPACE$ at 209 59
  \pinlabel $y$ at 168 52
  \pinlabel $PH$ at 121 70
\endlabellist
\centering
\includegraphics[scale=1.5]{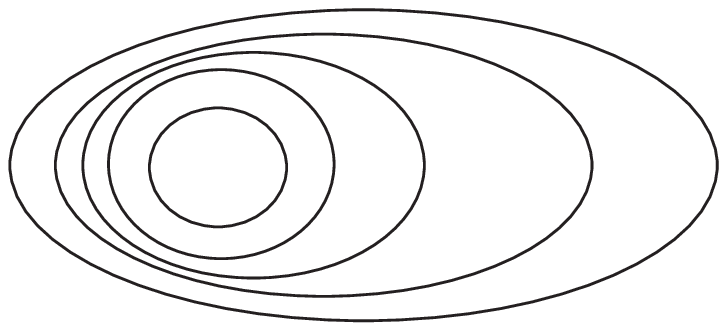}
\caption{} \label{cplxincl}
\end{figure}

Briefly, $P$ consists of decision (yes/no) problems or languages for which membership is determined in polynomial time (in input size) on a classical computer (Turing machine).  $NP$ (nondeterministic polynomial) is the class of languages which have a polynomial time protocol such that ``yes'' instances have a certificate which is accepted whereas there is no such requirement for ``no'' instances. $\# P$ is the counting analogy to $NP$ and asks how many of a fixed family of potential certificates will be accepted; the paradigmatic example problem being to find the number of assignments satisfying a boolean
formula. Since $\# P$ is a class of functions, not languages, one sometimes weakens $\# P$ to class $PP$ of languages where membership is determined by asking if more than half of the nondeterministic computations are accepting. $PP$ ``sees'' the first bit of $\# P$. We use the oracle notation $P^A$ in the sense of Cook (also called ``Turing reduction''), to mean polynomial time computation assisted by (possibly repeated) calls to the $A$ oracle (post processing permitted).  It is known that $P^{PP} = P^{\# P}$, so weakening $\#
P$ to a language does not affect its oracular power.  A function $f$ is called $\# P$-hard if $P^{\# P} \subseteq P^{A}$, $A$ an oracle for $f$.  $PH$ denotes the polynomial time hierarchy, a game theoretic extension of $NP$ allowing finite quantification. Toda proved that $PH \subseteq P^{PP}$ \cite{Toda}. Finally, $PSPACE$ is the class of decision problems solvable in an arbitrary amount of time, but using only a polynomial memory resource.  See \cite{Papa} for more background.

We use Axiom A, $P^{\# P} \neq NP$, to prove Theorem A.  Failure of
Axiom A would imply a large collapse of the polynomial hierarchy
$PH$ down to $NP$, so Axiom A must be considered extremely safe.

\section{Axiom A Implies Theorem A}

Our connection between links $L$ and complexity is the Jones polynomial \cite{Turaev} which we write as $J_L(q)$. Evaluations of $J_L$ at roots of unity $\omega = e^{2 \pi i/r}$ are known \cite{RT} to be computed as the partition function $Z_{SU(2),k}(S^3, L)$ of the topological quantum field theory (TQFT) associated with the Lie group $SU(2)$ at level $k=r-2$.  What will be of critical importance for us is that these Jones evaluations $J_L(\omega)$ will be constant as $L$ is transformed to $L' \sim_r L$.

\begin{lemma} \label{J_L_lemma}
If $L \sim_r L'$ then $J_L(e^{2\pi i/r}) = J_{L'}(e^{2\pi i/r})$.
\end{lemma}

\begin{proof}
In the $SU(2)_{r-2}$ theories, all ``labels'' $a$ (that is, positive
normed irreps of the quantum group, or ``particle types'' in physics
language) have twist factor $\theta(a)$ which is a $4r$-th root of
unity. Specifically, enumerating $a=0,\dots,r-2$, one has $\theta(a)
= \beta^{a^2+2a}$ where $\beta = e^{2 \pi i/4r}$ \cite{RT}.

\begin{figure}[]
\labellist \small\hair 2pt
  \pinlabel $a$ at 79 4
  \pinlabel $a$ at 159 4
  \pinlabel $\theta(a)$ at 147 50
  \pinlabel $=$ at 120 50
  \pinlabel $\theta(a)^{4r}=1$ at 220 50
\endlabellist
\centering
\includegraphics[scale=1]{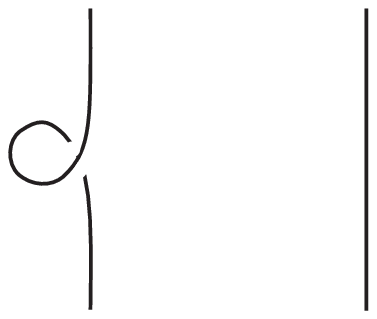}
\caption{} \label{twist}
\end{figure}

Now consider $L \coprod U$ where $U$ is a single unknotted loop bounding an imbedded disk $\Delta$ transverse to $L$.  Recoupling transforms $L$ to a superposition of trivalent ribbon graphs $\sum \alpha_i G_i$ with identical partition function, where each $G_i$ meets $\Delta$ in one edge with label $a_i$.  Now the partition function $Z(S^3,L) = J_L(e^{2\pi i/r})$ can be computed as $\sum \alpha_i Z(G_i)$.  But passing from $L$ to $L'$ amounts only to adding $4r$ full twists of the type drawn in figure \ref{twist} to the $a_i$ labeled particle line crossing $\Delta$.  Since $\theta(a_i)^{4r} =
1$, $Z(G_i)$ does not change under a $8\pi r$ twist. Consequently, $J_L(e^{2\pi i/r})=J_{L'}(e^{2\pi i/r})$.  I thank Ian Agol for pointing out that Fox \cite{Fox} considered a relation similar to $\sim_r$ in the 1950's and that Lackenby's theorem 2.1 \cite{Lackenby} contains lemma \ref{J_L_lemma}.
\end{proof}

It is a theorem of Vertigan (\cite{Vertigan2} or \cite{Vertigan1}
assisted by the result of \cite{Thistle}) that all non-zero
algebraic evaluations of the Jones polynomial $J_L(q)$ are $\#
P$-hard functions\footnote{Actually, applying Lagrangian interpolation, these functions are shown in \cite{JVW} to be $FP^{\# P}$-complete} of the input $L$ with the exceptions of those
$q$ satisfying $q^4=1$ or $q^6=1$. Thus, in oracle notation,
$P^{\J_r} = P^{\# P}$ where $\J_r$ accepts $L$ as input and returns
(an encoding of the algebraic integer) $J_L(e^{2\pi i/r})$, provided
$r \geq 5$ is an integer and $r \neq 6$.

From the lemma we see that the oracle $\J_r$ can work equally well with any link $L' \sim_r L$ as input or any diagram $D'$ for $L'$. But if $g(D') \leq b\log(c(L))+b'$, then the ``physical'' Hilbert space (i.e. the Hilbert space associated by $SU(2)_k$ TQFT to the $z=$ constant slices of $L$ (with charge
$a=1=$ fundamental)) will throughout the computation of the partition function have dimension $\displaystyle d < \sum_{i=0}^{r-2} S_{0,i}^{-(b \log c(D(L))+b')} < poly(c(D))$, using the Verlinde formula (VF), where $S_{0,i}=\sqrt{\frac{2}{r}} \sin(\frac{(i+1)\pi}{r})$, the first row of the $S$-matrix.  We have used minus our bound on girth as a lower bound to the Euler characteristic (the exponent in VF) for any $z=$ constant slice of the link complement in $\R^3$.

Thus, there is a prospect of replacing the oracle $\J_r$ entirely with a classical polynomial time computation in this small Hilbert space, by representing crossings by $R$-matrices and maxima (minima) by (co)units (as in Turaev's book \cite{Turaev}). To do this, two things must happen. First, $c(D')$ cannot be larger than $\text{poly}(c(D))$, that is,
the diagram $D'$, although fairly thin, also must not be too long in the $z$-direction. Second, there must be a polynomial number of advice bits which encode the steps from $D$ to $D'$ which certify that $D' \sim_r D$. If Theorem A were false, these poly-many advice bits could be used to certify transformations $D \sim_r D'$ where $D'$ would be thin enough, $g(D') < b \log(c(L))+b'$ and short enough $c(D') < c(D) + p(c(D))$ for a poly-time calculation of $J_{D'}(e^{2\pi i/r})$ to replace appeal to the oracle $\J_r$ implying $P^{\J_r} \subset NP$, contradicting Axiom A. We have used that $dist_r(D, D') < p(c(D))$ implies $c(D') < c(D) + p(c(D))$ since no more than two crossings or two critical points can be added to a diagram per unit weight step. This completes the proof of Theorem A in $ZF \cup \text{Axiom A}$.

\section{Conclusion}

Mathematical structures such as tilings \cite{Berger}, groups
\cite{Stillwell}, and, in several contexts, links \cite{FKLW} are
known to encode quite general computations. If transformations are
found which preserve the computational ``content'' of the structure,
then it may be expected that axioms stating a lower bound to
computational complexity will limit the scope of such
transformations in simplifying the structure.

\bibliographystyle{plain}
\bibliography{ccaxiom}

\begin{thebibliography}{10}

\bibitem{Berger}
Robert Berger.
\newblock The undecidability of the domino problem.
\newblock {\em Mem. Amer. Math. Soc. No.}, 66:72, 1966.

\bibitem{Fox}
R.~H. Fox.
\newblock Congruence classes of knots.
\newblock {\em Osaka Math. J.}, 10:37--41, 1958.

\bibitem{FKLW}
Michael~H. Freedman, Alexei Kitaev, Michael~J. Larsen, and Zhenghan Wang.
\newblock Topological quantum computation.
\newblock {\em Bull. Amer. Math. Soc. (N.S.)}, 40(1):31--38 (electronic), 2003.
\newblock Mathematical challenges of the 21st century (Los Angeles, CA, 2000).

\bibitem{Friedman}
Harvey~M. Friedman.
\newblock Finite functions and the necessary use of large cardinals.
\newblock {\em Ann. of Math. (2)}, 148(3):803--893, 1998.

\bibitem{JVW}
F.~Jaeger, D.~L. Vertigan, and D.~J.~A. Welsh.
\newblock On the computational complexity of the {J}ones and {T}utte
  polynomials.
\newblock {\em Math. Proc. Cambridge Philos. Soc.}, 108(1):35--53, 1990.

\bibitem{Lackenby}
Marc Lackenby.
\newblock Fox's congruence classes and the quantum-{${\rm SU}(2)$} invariants
  of links in {$3$}-manifolds.
\newblock {\em Comment. Math. Helv.}, 71(4):664--677, 1996.

\bibitem{Nabutovsky}
Alexander Nabutovsky.
\newblock Non-recursive functions, knots ``with thick ropes'', and
  self-clenching ``thick'' hyperspheres.
\newblock {\em Comm. Pure Appl. Math.}, 48(4):381--428, 1995.

\bibitem{Papa}
Christos~H. Papadimitriou.
\newblock {\em Computational complexity}.
\newblock Addison-Wesley Publishing Company, Reading, MA, 1994.

\bibitem{Stillwell}
John Stillwell.
\newblock The word problem and the isomorphism problem for groups.
\newblock {\em Bull. Amer. Math. Soc. (N.S.)}, 6(1):33--56, 1982.

\bibitem{Thistle}
Morwen~B. Thistlethwaite.
\newblock A spanning tree expansion of the {J}ones polynomial.
\newblock {\em Topology}, 26(3):297--309, 1987.

\bibitem{Thompson}
A.~Thompson.
\newblock Private communication.

\bibitem{Toda}
S.~Toda.
\newblock On the computational power of $pp$ and $\oplus p$.
\newblock {\em Proc. 30th IEEE Symposium on the Foundations of Computer
  Science}, pages 514--519, 1989.

\bibitem{Turaev}
V.~G. Turaev.
\newblock {\em Quantum invariants of knots and 3-manifolds}, volume~18 of {\em
  de Gruyter Studies in Mathematics}.
\newblock Walter de Gruyter \& Co., Berlin, 1994.

\bibitem{RT}
V.~G. Turaev and N.~Reshetikhin.
\newblock Invariants of {$3$}-manifolds via link polynomials and quantum
  groups.
\newblock {\em Invent. Math.}, 103(3):547--597, 1991.

\bibitem{Vertigan2}
Dirk Vertigan.
\newblock The computational complexity of {T}utte, {J}ones, {H}omfly and
  {K}aufman invariants.
\newblock DPhil Thesis, Oxford University, Oxford, England, 1991.

\bibitem{Vertigan1}
Dirk Vertigan.
\newblock The computational complexity of {T}utte invariants for planar graphs.
\newblock {\em SIAM J. Comput.}, 35(3):690--712 (electronic), 2005.

\end{thebibliography}

\end{document}